\documentclass[a4paper,onecolumn,11pt,unpublished]{quantumarticle}
\pdfoutput=1

\usepackage[utf8]{inputenc}
\usepackage[english]{babel}
\usepackage[T1]{fontenc}
\usepackage{amsmath,amssymb,amsfonts,amsthm,mathtools}
\usepackage[numbers,sort&compress]{natbib}
\usepackage{microtype}
\usepackage{hyperref}
\usepackage{graphicx}
\usepackage{float}
\usepackage{placeins}
\hypersetup{
  colorlinks=true,
  citecolor=blue!55!black,
  linkcolor=blue!55!black,
  urlcolor=blue!55!black
}

\newtheorem{theorem}{Theorem}
\newtheorem{lemma}{Lemma}

\newtheorem{corollary}{Corollary}

\DeclareMathOperator{\Tr}{Tr}
\DeclareMathOperator{\supp}{supp}

\newcommand{\E}{\mathbb{E}}
\newcommand{\Prb}{\mathbb{P}}
\newcommand{\bits}{\{0,1\}}
\newcommand{\dist}{d_{\mathrm H}}
\newcommand{\avg}{\mathrm{avg}}
\newcommand{\opt}{\mathrm{opt}}

\newcommand{\PGM}{\mathrm{PGM}}

\AtBeginDocument{\normalfont}

\begin{document}

\title{Optimal Average Success Probabilities of Binary $(n,n-1)$ and $(n,n-2)$ Quantum Random Access Codes}

\author{Shuo Tan}
\email{shuot8@uci.edu}
\author{Syed A. Jafar}
\email{syed@uci.edu}

\affiliation{Department of Electrical Engineering and Computer Science,\\
University of California, Irvine, California 92697, USA}

\maketitle

\hypersetup{
  pdftitle={Optimal Average Success Probabilities of Binary (n,n-1) and (n,n-2) Quantum Random Access Codes},
  pdfauthor={Shuo Tan and Syed A. Jafar},
  pdfsubject={Quantum random access codes and the positive semidefinite structure of the joint probability matrix induced by the full-string PGM},
  pdfkeywords={quantum random access codes, pretty good measurement, positive semidefinite matrix, linear programming, Krawtchouk polynomials}
}

\begin{abstract}
A binary $(n,m)$ quantum random access code (QRAC) compresses an $n$-bit classical string into an $m$-qubit quantum state, from which a decoder attempts to recover a randomly selected target bit. Of particular interest is the optimal average probability of success, $P^{Q,\avg,\opt}_{n,m}$, which is numerically conjectured to satisfy the bound $P^{Q,\avg,\opt}_{n,m}\leq \frac{1}{2}+\frac{1}{2}\sqrt{\frac{m}{n}}$. Recent constructions of $(n,n-1)$ QRACs by Suzuki and $(n,n-2)$ QRACs by Akibue et al. meet this bound exactly, raising the question of their strict optimality. In this work, we settle this question by proving the conjectured upper bound for $m\in\{n-1,n-2\}$, thereby  precisely determining $P^{Q,\avg,\opt}_{n,n-1}$ and $P^{Q,\avg,\opt}_{n,n-2}$. The proof utilizes a translation recently studied by Lin and de Wolf from local to global reconstruction via pretty good measurement, along with dimensional and positive-semidefinite constraints on an induced channel.
\end{abstract}

\section{Introduction}
\label{sec:introduction}

The fundamental limits of data compression are a central theme in quantum information theory. Though Holevo's theorem \cite{Holevo1973} shows that a system of $m$ qubits cannot be used to retrieve more than $m$ bits of classical information, this global limitation does not fully capture how reliably individual bits of a classical string can be accessed when $n$ bits are compressed into $m$ qubits.   Quantum random access codes (QRACs), introduced by Ambainis et al.~\cite{AmbainisNayakTaShmaVazirani1999}, formalize this tradeoff between compressibility and retrievability. In an $(n,m)$ binary QRAC, a classical string $x\in\bits^n$ is encoded into an $m$-qubit state $\rho_x$, from which a decoder attempts to recover a randomly selected target bit as reliably as possible. To avoid degenerate cases we assume throughout that $0<m<n$. A key performance metric is the optimal average success probability, $P^{Q,\avg,\opt}_{n,m}$.

The conceptual significance of QRACs touches upon a variety of problems, including the study of quantum correlations through the principle of information causality~\cite{pawlowski2009information}, lower bounds in one-way quantum communication complexity~\cite{Klauck2000,Aaronson2004} and quantum network coding~\cite{HayashiEtAl2007}, as well as bounds for locally decodable codes~\cite{KerenidisDeWolf2004} and private information retrieval~\cite{WehnerDeWolf2005}. They also arise in dimension witnessing~\cite{BrunnerNavascuesVertesi2013}, semi-device-independent cryptography~\cite{PawlowskiBrunner2011,LiEtAl2011,LiPawlowskiYinGuoHan2012}, high-dimensional communication~\cite{TavakoliHameediMarquesBourennane2015}, and self-testing~\cite{FarkasKaniewski2019}. All of these  motivate a precise characterization of $P^{Q,\avg,\opt}_{n,m}$.

\subsection{Prior work and problem status}

While much progress has been made through various bounds and QRAC constructions~\cite{Nayak1999,AmbainisNayakTaShmaVazirani1999,AmbainisLeungMancinskaOzols2008,ImamichiRaymond2018,FarkasKaniewski2019,MancinskaStorgaard2022,FarkasMiklinTavakoli2025,LinDeWolf2025,Suzuki2026,AkibueRaymondTamakiTeramoto2026,KondoEtAl2026}, exact characterizations of $P^{Q,\avg,\opt}_{n,m}$ were previously known only in a handful of cases. The values for $(n,m)=(2,1),(3,1)$ follow from~\cite{AmbainisNayakTaShmaVazirani1999,AmbainisLeungMancinskaOzols2008}. For $(n,m)=(3,2),(4,2),(6,2)$, the constructions of Imamichi and Raymond~\cite{ImamichiRaymond2018} meet the $m=2$ converse of Man\v{c}inska and Storgaard~\cite{MancinskaStorgaard2022}. Thus, 
\begin{align}
P^{Q,\avg,\opt}_{n,m}
&=\frac12+\frac12\sqrt{\frac mn},
&
(n,m)&\in\{(2,1),(3,1),(3,2),(4,2),(6,2)\}.
\label{eq:known-exact}
\end{align}
Remarkably, the same form also appears in a numerically conjectured general upper bound noted  in \cite{Suzuki2026,AkibueRaymondTamakiTeramoto2026}: 
\begin{equation}
P^{Q,\avg,\opt}_{n,m}
\stackrel{?}{\leq}
\frac12+\frac12\sqrt{\frac mn}.
\label{eq:conjbound}
\end{equation}
It is known that the bound \eqref{eq:conjbound} holds if $m\in\{1,2\}$, because the bound $P^{Q,\avg,\opt}_{n,m}\leq \frac{1}{2}+\frac{1}{2}\sqrt{\frac{2^{m-1}}{n}}$ is established in \cite{MancinskaStorgaard2022}, and for $m\in\{1,2\}$ we have $2^{m-1}=m$. 

Although Liu's classical random access code (RAC) counterexamples rule out \eqref{eq:conjbound} as a universal upper bound~\cite{Liu2026}, the high-rate cases $m\in\{n-1,n-2\}$ remain open. These cases are particularly of interest because known QRAC constructions attain the square-root value. Specifically the $(n,n-1)$ QRACs constructed by Suzuki in \cite{Suzuki2026}, and the $(n,n-2)$ QRACs constructed by Akibue et al. in \cite{AkibueRaymondTamakiTeramoto2026}, achieve average success probability $\frac{1}{2}+\frac{1}{2}\sqrt{\frac{m}{n}}$ for $m=n-1$ and $m=n-2$, respectively. Thus, if the conjectured  bound \eqref{eq:conjbound} can be proved for  $m\in\{n-1,n-2\}$, then it would prove the optimality of both constructions, and thus determine the precise values of $P^{Q,\avg,\opt}_{n,n-1}$ and $P^{Q,\avg,\opt}_{n,n-2}$. Indeed, a proof of \eqref{eq:conjbound} for $m\in\{n-1,n-2\}$ is the main contribution of this work.

\subsection{Main results and proof idea}

Our main result answers this question and determines the optimum for both families:
\begin{equation}
P^{Q,\avg,\opt}_{n,m}
=\frac12+\frac12\sqrt{\frac mn},
\qquad m\in\{n-1,n-2\}.
\label{eq:intro-exact-high-rate}
\end{equation}
Our converse builds on the pretty good measurement (PGM) reconstruction argument of Lin and de Wolf \cite{LinDeWolf2025}.  The full-string PGM \cite{HausladenWootters1994, HausladenEtAl1996} produces an entire candidate string \(Y^{\mathfrak{C},\PGM}\) for the  uniform string \(X\) encoded by a QRAC $\mathfrak{C}$.  Lin and de Wolf  in \cite[Theorem 2]{LinDeWolf2025} show that  reducing the outcome of the full-string PGM to the $i$th bit by summing over the possible outcomes for other bits induces exactly the one-bit PGM for the $i$th bit.  

Building on this foundation, we find upper and lower bounds on the expected Hamming distance $\Delta^{\mathfrak{C},\PGM}:=\E[\dist(X,Y^{\mathfrak{C},\PGM})]$. For the upper bound,  if a given QRAC $\mathfrak{C}$ has average success probability
\(P^{\mathfrak{C}}\), then by averaging over one-bit PGMs we show that $\Delta^{\mathfrak{C},\PGM} \leq \frac{n}{2}\left(1-\left(2P^{\mathfrak{C}}-1\right)^2\right)$. For the lower bound, based on the full-string PGM we show that when $m\in\{n-1,n-2\}$, then  $\Delta^{\mathfrak{C},\PGM}\geq \frac{n-m}{2}$.  Combining these bounds  produces the desired proof of \eqref{eq:conjbound} with $m\in\{n-1,n-2\}$ for the QRAC $\mathfrak{C}$. Since the bound holds for any $(n,m)$ QRAC $\mathfrak{C}$ with $m\in\{n-1,n-2\}$, it also holds for the optimal average success probability.

\section{QRACs and the Full-String PGM}
\label{sec:model-pgm}

\subsection{Average-case QRAC model}
\label{sec:model}

Let $[n]=\{1,\ldots,n\}$. A binary $(n,m)$ QRAC $\mathfrak{C}$ is specified as a tuple $(\{\rho_x^{\mathfrak{C}}\}_{x\in\{0,1\}^n},\allowbreak \{M^{\mathfrak{C}}_{0|i},M^{\mathfrak{C}}_{1|i}\}_{i\in[n]})$, where $\{\rho^{\mathfrak{C}}_x\}_{x\in\{0,1\}^n}$ is an ensemble of quantum states on a Hilbert space of dimension at most $2^m$, and  $\{M^{\mathfrak{C}}_{0|i},M^{\mathfrak{C}}_{1|i}\}$ is a binary positive operator-valued measure (POVM) for decoding the $i$th bit. The encoder obtains a uniformly random binary string $X\in\{0,1\}^n$ and maps any realization $X=x=x_1x_2\cdots x_n$ to the corresponding quantum state $\rho_x^{\mathfrak{C}}$. The decoder is given the encoded quantum state, along with an independently generated uniformly random index $\mathcal{I}\in[n]$. Given $\rho^{\mathfrak{C}}_x$ and the realization $\mathcal{I}=i$, the decoder produces $Y_i^{\mathfrak{C}}$ as the measurement outcome of the POVM $\{M^{\mathfrak{C}}_{0|i},M^{\mathfrak{C}}_{1|i}\}$ applied to $\rho^{\mathfrak{C}}_x$. Define the coordinate-wise success probability of $\mathfrak{C}$ as
\begin{align}
p_i^{\mathfrak{C}}&:=\Prb(Y_i^{\mathfrak{C}}=X_i\mid \mathcal{I}=i)\notag\\
&= \frac{1}{2^n}\sum_{x\in\{0,1\}^n}\Prb(Y_i^{\mathfrak{C}}=X_i\mid X=x, \mathcal{I}=i)\notag\\
&=\frac{1}{2^n}\sum_{x\in\{0,1\}^n}\Tr\!\left(\rho^{\mathfrak{C}}_x M^{\mathfrak{C}}_{x_i|i}\right),\label{eq:coordinate-success}\\
\intertext{and the average success probability (ASP) of  $\mathfrak{C}$ as}
P^{\mathfrak{C}}&:=\frac{1}{n}\sum_{i=1}^n p^{\mathfrak{C}}_i.
\end{align}
Define the coordinate-wise success probability of $\mathfrak{C}$, optimized over all  POVMs for decoding the $i$th bit, as
\begin{align}
p_i^{\mathfrak{C}*}&:= \max_{\mbox{\tiny POVM}:\{M_{0|i},M_{1|i}\}}\frac{1}{2^n}\sum_{x\in\{0,1\}^n}\Tr\!\left(\rho^{\mathfrak{C}}_x M_{x_i|i}\right)\label{eq:POVMopt}
\end{align}
The \emph{optimal} ASP for the entire class of $(n,m)$ QRACs, $P^{Q,\avg,\opt}_{n,m}$ is defined as in \cite{Suzuki2026},
\begin{align}
&P^{Q,\avg,\opt}_{n,m}
    :=\max_{\mathfrak{C}}P^{\mathfrak{C}}\\
&=    \max_{\{\rho_x\},\{M_{b|i}\}}
    \frac{1}{n2^n}
    \sum_{x\in\bits^n}
    \sum_{i=1}^n
    \Tr\!\left(\rho_x M_{x_i|i}\right).
    \label{eq:opt-avg-success}
\end{align}

\subsection{The full-string PGM and Hamming distortion}
\label{sec:pgm}

The full-string PGM is defined from the average state
\begin{align}
\bar\rho^{\mathfrak{C}}:=2^{-n}\sum_{x\in\{0,1\}^n}\rho^{\mathfrak{C}}_x
    \label{eq:average-state}
\end{align}
as in \cite{LinDeWolf2025} with uniform $X$ as follows.
\begin{equation}
    Q_y^{\mathfrak{C}}=2^{-n}(\bar\rho^{\mathfrak{C}})^{-1/2}\rho^{\mathfrak{C}}_y(\bar\rho^{\mathfrak{C}})^{-1/2},
    \qquad y\in\bits^n,
    \label{eq:pgm-def}
\end{equation}
where inverse powers are taken on $\supp(\bar\rho^{\mathfrak{C}})$ and are zero on the
orthogonal complement.  Then 
$\sum_y Q^{\mathfrak{C}}_y=\Pi_{\supp(\bar\rho^{\mathfrak{C}})}$ is the orthogonal projection matrix onto $\supp(\bar\rho^{\mathfrak{C}})$, and 
$\{Q^{\mathfrak{C}}_y\}_y$ is a POVM on $\supp(\bar\rho^{\mathfrak{C}})$; if desired as a POVM on the whole ambient Hilbert space, it may be completed arbitrarily on
$\supp(\bar\rho^{\mathfrak{C}})^\perp$. This completion does not affect the probabilities
below, because all code states are supported on $\supp(\bar\rho^{\mathfrak{C}})$.  Let $Y^{\mathfrak{C},\PGM}\in\{0,1\}^n$ denote the PGM output when the input is $X\in\{0,1\}^n$. 

For $x,y\in\bits^n$, write
$$\dist(x,y)=\bigl|\{i\in[n]:x_i\ne y_i\}\bigr|$$ for their Hamming distance.  For the random pair \((X,Y^{\mathfrak{C},\PGM})\), define the PGM distortion by
\begin{align}
    \Delta^{\mathfrak{C},\PGM}:=\E[\dist(X,Y^{\mathfrak{C},\PGM})].
\end{align}

The full-string PGM induces a classical channel with transition probabilities $T_{xy}^{\mathfrak{C},\PGM}$, where $x\in\{0,1\}^n$ is the input and $y\in\{0,1\}^n$ is the output:
\begin{equation}
    T_{xy}^{\mathfrak{C},\PGM}=\Prb[Y^{\mathfrak{C},\PGM}=y\mid X=x]=\Tr(\rho^{\mathfrak{C}}_xQ^{\mathfrak{C}}_y).
    \label{eq:T-def}
\end{equation}

\section{Main Results}
\label{sec:exact-high-rate}

We start with three useful lemmas. The first lemma produces an upper bound on the expected Hamming distance, $\Delta^{\mathfrak{C},\PGM}$, between $X$ and the full-string PGM output $Y^{\mathfrak{C},\PGM}$.
\begin{lemma}%[Average-success PGM bridge]
\label{lem:Deltaup}
For the full-string PGM output \(Y^{\mathfrak{C},\PGM}\),
\begin{align}
%\Delta^{\mathfrak{C},\PGM} \le 2nP^{\mathfrak{C}}(1-P^{\mathfrak{C}}).
\Delta^{\mathfrak{C},\PGM} \leq\frac{n}{2}\left(1-(2P^{\mathfrak{C}}-1)^2\right).
\label{eq:Deltaup}
\end{align}
\end{lemma}

\begin{proof}
From Lin and de Wolf \cite[Proof of Theorem~2]{LinDeWolf2025}, let us recall that reducing the outcome of the full-string PGM to the $i$th bit by summing over the possible outcomes for other bits induces exactly the one-bit PGM for the $i$th bit. Thus, the $i$th coordinate of the full-string PGM output $Y^{\mathfrak{C},\PGM}$, namely  $Y_i^{\mathfrak{C},\PGM}$, represents the output of a one-bit PGM for the $i$th bit from the  binary marginal ensemble
\(\{(1/2,\sigma^{\mathfrak{C}}_{i,0}),(1/2,\sigma^{\mathfrak{C}}_{i,1})\}\), where
\begin{align}
\sigma^{\mathfrak{C}}_{i,b}:=\frac{1}{2^{n-1}}\sum_{x\in\{0,1\}^n: x_i=b}\rho_x^{\mathfrak{C}},\qquad b\in\{0,1\}.
\end{align}
This allows us to apply Renes's refined PGM bound \cite[Eq.~(12)]{Renes2017} for a \emph{one-bit} PGM to obtain,
\begin{align}
\Prb[Y_i^{\mathfrak{C},\PGM}=X_i]&\geq (p_i^{\mathfrak{C}*})^2+(1-p_i^{\mathfrak{C}*})^2\label{eq:renes}
\end{align}
The optimized POVM in \eqref{eq:POVMopt} is at least as good as choosing the better of the  POVM prescribed by $\mathfrak{C}$ and its swapped version, i.e.,
\begin{align}
p_i^{\mathfrak{C}*}\geq \max\{p_i^{\mathfrak{C}},1-p_i^{\mathfrak{C}}\}\geq 1/2.\label{eq:swap}
\end{align}
Recognizing that the function $z^2+(1-z)^2$ is increasing for $z\geq 1/2$, we combine \eqref{eq:renes} and \eqref{eq:swap} to obtain,
\begin{align}
&\Prb[Y_i^{\mathfrak{C},\PGM}=X_i]\notag\\
&\geq (\max\{p_i^{\mathfrak{C}},1-p_i^{\mathfrak{C}}\})^2+(1-\max\{p_i^{\mathfrak{C}},1-p_i^{\mathfrak{C}}\})^2\\
&= (p_i^{\mathfrak{C}})^2+(1-p_i^{\mathfrak{C}})^2\label{eq:Crenes}\\
\intertext{which is equivalently expressed as}
&\Prb[Y_i^{\mathfrak{C},\PGM}\neq X_i]\leq 2p_i^{\mathfrak{C}}(1-p_i^{\mathfrak{C}}).
\end{align}
Summing over all $i\in[n]$ yields,
\begin{align}
\Delta^{\mathfrak{C},\PGM}&=\E[\dist(X,Y^{\mathfrak{C},\PGM})]\notag\\
&=\sum_{i\in[n]}\Prb[Y_i^{\mathfrak{C},\PGM}\neq X_i]\\
&\leq \sum_{i\in[n]}2p_i^{\mathfrak{C}}(1-p_i^{\mathfrak{C}})\\
&\leq 2nP^{\mathfrak{C}}(1-P^{\mathfrak{C}}).\label{eq:Jensens}
\end{align}
Step \eqref{eq:Jensens} follows from an application of Jensen's inequality, recognizing that $z(1-z)$ is a concave function of $z$. Finally, rearranging the bound $\Delta^{\mathfrak{C},\PGM}\leq 2nP^{\mathfrak{C}}(1-P^{\mathfrak{C}})$ yields the desired bound \eqref{eq:Deltaup}.

\end{proof}
The second lemma establishes that the matrix of transition probabilities for the channel $T^{\mathfrak{C},\PGM}$ defined in \eqref{eq:T-def} is positive semidefinite (PSD).
\begin{lemma}
\label{lem:channel-structure}
The matrix $T^{\mathfrak{C},\PGM}$ is 
positive semidefinite.  
\end{lemma}

\begin{proof}
Define
\begin{align}
    A_x^{\mathfrak{C}}:=(\bar\rho^{\mathfrak{C}})^{-1/4}\rho_x^{\mathfrak{C}}(\bar\rho^{\mathfrak{C}})^{-1/4} 
\end{align}
and note that the cyclicity of trace implies
\begin{align}
    T_{xy}^{\mathfrak{C},\PGM}&=\Tr(\rho^{\mathfrak{C}}_xQ^{\mathfrak{C}}_y)\\
    &=\Tr(\rho^{\mathfrak{C}}_x2^{-n}(\bar\rho^{\mathfrak{C}})^{-1/2}\rho^{\mathfrak{C}}_y(\bar\rho^{\mathfrak{C}})^{-1/2})\\
 &=2^{-n}\Tr(A_x^{\mathfrak{C}}A_y^{\mathfrak{C}})
    \label{eq:T-gram}
\end{align}
The matrix $T^{\mathfrak{C},\PGM}=(2^{-n}\Tr(A_x^{\mathfrak{C}}A_y^{\mathfrak{C}}))_{x,y}$ is a real symmetric Gram matrix in the Hilbert-Schmidt inner product, which makes it positive semidefinite. 

\end{proof}

The third lemma bounds the probability that the Hamming distance between $X$ and the full-string PGM output is odd, essentially by showing that for any $n$-bit random strings $X$ and $Y$, if the matrix representation of their joint probability mass function is positive semidefinite, then the Hamming distance between $X$ and $Y$ is not more likely to be odd than it is to be even.
\begin{lemma}
\label{lem:parity-constraint}
We have the following bound,
\begin{equation}
    \Prb[\dist(X,Y^{\mathfrak{C},\PGM})\text{ is odd}]\le \frac{1}{2}.
    \label{eq:odd-bound}
\end{equation}

\end{lemma}
\begin{proof}

Let $s: \{0,1\}^n \to \{-1, 1\}$ be defined by $s(x) = (-1)^{|x|}$, where $|x|$ denotes the Hamming weight of $x$.  Using the fact established in Lemma \ref{lem:channel-structure}, i.e.,
$T^{\mathfrak{C},\PGM}\succeq 0$, we have
\begin{align}
    0
    &\le 2^{-n}\sum_{x,y}s(x)T_{xy}^{\mathfrak{C},\PGM}s(y)\\
    &=2^{-n}\sum_{x,y}s(x)\mathbb{P}[Y^{\mathfrak{C}, \PGM}=y\mid X=x]s(y)\label{eq:useTdef}\\
   &=\sum_{x,y}s(x)\mathbb{P}[Y^{\mathfrak{C},\PGM}=y, X=x]s(y)\label{eq:Xuni}\\
    &= \E[s(X)s(Y^{\mathfrak{C},\PGM})]\\
    &= \E[(-1)^{|X|+|Y^{\mathfrak{C},\PGM}|}]\label{eq:uses}\\
    &= \E[(-1)^{\dist(X,Y^{\mathfrak{C},\PGM})}]\label{eq:ss2H}\\
    &= \Prb[\dist(X,Y^{\mathfrak{C},\PGM})\text{ is even}]\notag\\
&\hspace{1cm}       -\Prb[\dist(X,Y^{\mathfrak{C},\PGM})\text{ is odd}]\\
&=1-2\Prb[\dist(X,Y^{\mathfrak{C},\PGM})\text{ is odd}]\label{eq:Hodd}
\end{align}
Step \eqref{eq:useTdef} follows from \eqref{eq:T-def}. Step \eqref{eq:Xuni} follows from the uniform distribution on $X$. Step \eqref{eq:ss2H} follows from the fact that $(|x|+|y|)\mod 2$  is equal to $\dist(x,y)\mod 2$. Rearranging  \eqref{eq:Hodd} yields  the desired bound \eqref{eq:odd-bound}.%
\end{proof}

The preceding lemmas yield the following converse.
\begin{theorem}
\label{thm:main-converse}
For \(m\in\{n-1,n-2\}\), every binary \((n,m)\) QRAC $\mathfrak{C}$  has average success probability bounded as
\begin{equation}
    P^{\mathfrak{C}}\le \frac12+\frac12\sqrt{\frac{m}{n}} .
    \label{eq:main-converse}
\end{equation}
Consequently Eq.~\eqref{eq:conjbound} holds for 
$P^{Q,\avg,\opt}_{n,m}$.
\end{theorem}

\begin{proof}
By relabeling the two outcomes of each coordinate POVM if necessary, we may
assume without loss of generality that $p_i^{\mathfrak{C}}\ge 1/2$ for all $i$.  This is because relabeling leaves the
encoding ensemble, and hence the full-string PGM channel and
\(\Delta^{\mathfrak{C},\PGM}=\E[\dist(X,Y^{\mathfrak{C},\PGM})]\), unchanged, while it can only increase the ASP. 

It will be useful to recall Nayak's result in \cite[Theorem 2.4, Part 2]{Nayak1999} which directly implies that $\mathbb{P}[Y=X]\leq 2^{m-n}$ if $Y\in\{0,1\}^n$ is obtained by making \emph{any} measurement (not restricted to PGM) of an $m$-qubit encoding of a uniform $X\in\{0,1\}^n$. In particular, this implies
\begin{align}
\Prb[Y^{\mathfrak{C},\PGM}=X]\le 2^{m-n}.\label{eq:genbound}
\end{align}

Case 1: For $m=n-1$ we have,
\begin{align}
\Delta^{\mathfrak{C},\PGM}&=\sum_{i\in[n]}\Prb[Y_i^{\mathfrak{C},\PGM}\neq X_i]\\
&\geq \mathbb{P}[Y^{\mathfrak{C},\PGM}\neq X]\\
&\geq 1/2\label{eq:applygenbound}
\end{align}
Step \eqref{eq:applygenbound} follows by applying \eqref{eq:genbound} with $m=n-1$. 

Case 2: Now consider $m=n-2$. For compact notation let us define $q_0:=\mathbb{P}[\dist(X,Y^{\mathfrak{C},\PGM})=0]$ and $q_{\rm odd}:=\mathbb{P}[\dist(X,Y^{\mathfrak{C},\PGM})\mbox{ is odd}]$.
\begin{align}
\Delta^{\mathfrak{C},\PGM}&=\E[\dist(X,Y^{\mathfrak{C},\PGM})]\\
&\geq \mathbb{P}[\dist(X,Y^{\mathfrak{C},\PGM})\mbox{ is odd}]\label{eq:evenoddbound}\\
&\hspace{1cm}+2\mathbb{P}[\dist(X,Y^{\mathfrak{C},\PGM})\mbox{ is nonzero and even}]\notag\\
&=q_{\rm odd}+2(1-q_0-q_{\rm odd})\\
&=2 - 2q_0- q_{\rm odd}\\
&\geq 2 - 2(1/4)-(1/2)\label{eq:useprevbounds}\\
&=1
\end{align}
Step \eqref{eq:evenoddbound} uses the fact that odd values of the Hamming distance must be at least one, and nonzero even values of the Hamming distance must be at least $2$. Step \eqref{eq:useprevbounds} applies the bound $q_0\leq 1/4$ which follows from \eqref{eq:genbound} for $m=n-2$, and the bound $q_{\rm odd}\leq 1/2$ which follows from Lemma \ref{lem:parity-constraint}.

Thus $\Delta^{\mathfrak{C},\PGM}\ge (n-m)/2$ in both cases.  Applying Lemma \ref{lem:Deltaup} and substituting this into \eqref{eq:Deltaup} yields
\begin{align}
   \frac{n-m}{2}&\leq \frac{n}{2}\left(1-(2P^{\mathfrak{C}}-1)^2\right)\\
   \implies n(2P^{\mathfrak{C}}-1)^2&\leq m .
\end{align}
Since $P^{\mathfrak{C}}=\frac{1}{n}\sum_{i=1}^n p^{\mathfrak{C}}_i\geq 1/2$ as assumed without loss of generality, for $m\in\{n-1,n-2\}$ we obtain the desired bound
\begin{align}
 P^{\mathfrak{C}}\le \frac12+\frac12\sqrt{\frac{m}{n}}.
\end{align}
\end{proof}
The exact values of $P^{Q,\avg,\opt}_{n,m}$ for $m\in\{n-1,n-2\}$ now follow as a corollary.
\begin{corollary}\label{cor:exact-values}
For \(m\in\{n-1,n-2\}\),
\begin{equation}
    P^{Q,\avg,\opt}_{n,m}
    =\frac12+\frac12\sqrt{\frac{m}{n}} .
    \label{eq:exact-high-rate}
\end{equation}
\end{corollary}
\begin{proof}
The upper bound follows from Theorem~\ref{thm:main-converse}.  For $m=n-1$,
Suzuki's $(n,n-1)$ construction achieves the right-hand side
\cite{Suzuki2026}.  For $m=n-2$ the $(n,n-2)$ construction of
Akibue et al. achieves the right-hand side 
\cite{AkibueRaymondTamakiTeramoto2026}.  
\end{proof}

\section{Discussion}
\label{sec:discussion}

We determined the precise values of the optimal average success probability for the classes of binary $(n,n-1)$ and $(n,n-2)$ QRACs by proving that the numerically conjectured upper bound \eqref{eq:conjbound} holds for these settings. Our proofs utilized  the properties identified in Lemmas $1,2,3$, of the classical PSD channel matrix induced by the full-string PGM. Limited extensions of the bounds are explored in \cite{Tan_QRAC_long}.

\section*{Acknowledgements}

The authors acknowledge the use of OpenAI's ChatGPT (GPT-5.5) as a research assistant during the discovery of the results presented in this work. Through an iterative, human-directed process, the authors formulated the research objectives, hypotheses, and problem constraints, while ChatGPT was used to generate intermediate calculations. All content is independently verified, revised, and approved by the authors, who take full responsibility for the manuscript.

\begingroup
\small
\bibliographystyle{quantum}
\bibliography{references}
\endgroup

\end{document}